\documentclass[conference]{IEEEtran}
\IEEEoverridecommandlockouts

\usepackage{cite}
\usepackage{amsmath,amssymb,amsfonts,amsthm,stmaryrd}
\usepackage{enumerate}
\usepackage{algorithm}
\usepackage{algorithmic}
\usepackage{graphicx}
\usepackage{textcomp}
\usepackage{xcolor}
\def\BibTeX{{\rm B\kern-.05em{\sc i\kern-.025em b}\kern-.08em
    T\kern-.1667em\lower.7ex\hbox{E}\kern-.125emX}}
\newtheorem{definition}{Definition}
\newtheorem{lemma}{Lemma}
\newtheorem{theorem}{Theorem}

\newtheorem{corollary}{Corollary}

\newcommand{\Ri}[0]{\mathbb{R}_{\geq 0}^{\infty}}

\begin{document}

\title{Bounding Distance Between Outputs in Distributed Lattice Agreement}

\author{\IEEEauthorblockN{Abdullah Rasheed}
\IEEEauthorblockA{\textit{The University of Texas at Austin}\\
Austin, TX, USA}
\and
\IEEEauthorblockN{Nidhi Dubagunta}
\IEEEauthorblockA{\textit{The University of Texas at Austin}\\
Austin, TX, USA}
}

\maketitle

\begin{abstract}
This paper studies the lattice agreement problem and proposes a stronger form, $\varepsilon$-bounded lattice agreement, that enforces an additional tightness constraint on the outputs. To formalize the concept, we define a quasi-metric on the structure of the lattice, which captures a natural notion of distance between lattice elements. We consider the bounded lattice agreement problem in both synchronous and asynchronous systems, and provide algorithms that aim to minimize the distance between the output values, while satisfying the requirements of the classic lattice agreement problem. We show strong impossibility results for the asynchronous case, and a heuristic algorithm that achieves improved tightness with high probability, and we test an approximation of this algorithm to show that only a very small number of rounds are necessary.

\end{abstract}


\section{Introduction}
The Lattice Agreement problem, introduced in \cite{attiya95}, is a weakened form of the distributed decision problem. Each process begins with a value from a lattice, and the goal is for all processes to decide on values that are mutually comparable. Lattice agreements loosens the constraints of the consensus problem, which requires strict uniformity and validity in the output values and is impossible in asynchronous systems \cite{flp}. Lattice agreement has the advantage of decidability in the presence of asynchrony and faults \cite{attiya95}, making it an important and useful problem in distributed systems.

Lattice agreement has a variety of applications in practice, and was first proposed as a method for attaining consistent atomic snapshots \cite{attiya95}. A generalized form of the lattice agreement problem was later proposed to implement fault-tolerant replicated state machines \cite{gla}. The generalized lattice agreement \cite{gla} is also shown to be useful in achieving linearizable data structures, particularly in conflict-free replicated data types. Another version of lattice agreement, reconfigurable lattice agreement, was proposed in \cite{rla} in order to handle agreement on the system configuration for client and server consistency. 

While better bounds on lattice agreement algorithms were found in \cite{bound}, they were improved upon in \cite{Zheng2018Lattice} with algorithms for synchronous, asynchronous, and generalized lattice agreement.

In this paper, we propose a new constraint, tightness, on standard lattice agreement that requires all outputs to be within a certain distance of each other as determined by a quasi-metric over the lattice. Such a requirement can force processes to agree on values (such as the global view in atomic snapshot) that are closer to each other. Minimizing this distance between intermediate states in applications such as atomic snapshots and replicated state machines could lead to faster eventual convergence. Reducing the divergence between replicas during the reconciliation phase may mean that fewer overall synchronization operations are needed, allowing the system to reach eventual consistency more efficiently. Additionally processes can potentially recover more quickly in the case of node crashes, because it reduces the risk of losing information from data loss. It can also ensure that processes are ``synchronized enough'' with respect to what they must agree on, in a similar notion as clock synchronization.

We show that this problem has equivalences to standard lattice agreement under particular bounds, and that there is a discrete structure to the range of tightness on the outputs, despite the continuity of the range on the quasi-metric. We briefly show that this problem is completely solved in the synchronous case, but is much trickier in the asynchronous case. While standard lattice agreement is solvable asynchronously, better tightness can never be guaranteed. This motivates our heuristic approach, which we provide an approximate model of and simulate to get approximate probabilities of improved tightness. We see that after only 5 rounds of our algorithm, the probability converges to a near guarantee. Finally, we conjecture that this number of rounds is constant, no matter the number of processes, faults, or the chance of a fault.

\section{Background}

We will let $[n]$ denote the set $\{1,\ldots,n\}$. We will also let $\Ri = \{r \in \mathbb{R} \mid r \geq 0\} \cup \{\infty\}$, where $r < \infty$ for all $r \in \Ri \setminus \{\infty\}$.

A triple $(X,\leq,\sqcup)$ is called a \emph{join semi-lattice} if $(X,\leq)$ is a poset, and $\sqcup$ is a join operator such that the join of every non-empty finite subset is defined. For brevity, we will refer to join semi-lattices simply as lattices. We say $X$ is \emph{grounded} if there exists an element $\bot_X \in X$ such that $\bot_X \leq x$ for all $x \in X$. For any subset $S \subseteq X$, we let $y \bowtie S$ mean $\exists s \in S : s \leq y \leq \bigsqcup S$.

With any lattice $X$, we will feel free to say $a < b$ for $a,b \in X$ to mean $a \leq b \wedge a \neq b$.

\section{System Model}

We consider a system model consisting of $n$ distributed processes, in which at most $f$ faults may occur. We only consider faults on processes, which suffer crash faults.

We will consider the system to be either synchronous or asynchronous, and each process can send and receive messages from every other process including itself.

\section{Distributed Lattice Agreement}

The original lattice agreement is as follows. Let $(X,\leq,\sqcup)$ be a finite lattice, and let there be $n$ processes in the system. Each process $p_i$ proposes an input value $x_i \in X$, and each $p_i$ decides on an output $y_i \in X$. To solve the lattice agreement problem, the following constraints must be met:
\begin{itemize}
    \item {\bf Downward-Validity}: $\forall i\in [n] : x_i \leq y_i$
    \item {\bf Upward-Validity}: $\forall i \in [n] : y_i \leq \bigsqcup_{j \in [n]}x_j$
    \item {\bf Comparability}: $\forall i,j \in [n] : y_i \leq y_j \vee y_j \leq y_i$
\end{itemize}

We introduce a stronger version of this problem by requiring the outputs to all be ``close'' to each other with respect to some distance function on the lattice.

\begin{definition}
    Let $(X,\leq,\sqcup)$ be a lattice. Then a function $\delta_X : X\times X \rightarrow \mathbb{R}_{\geq 0}^{\infty} \cup \{\bot\}$ is called a \emph{quasi-metric} over $X$ if for all $x_1,x_2,x_3 \in X$,
    \begin{enumerate}[(i)]
        \item $x_1 = x_2 \Leftrightarrow \delta_X(x_1,x_2) = 0$
        \item $x_1 \leq x_2 \Leftrightarrow \delta_X(x_1,x_2) \neq \bot$
        \item $x_1 \leq x_2 \leq x_3 \Rightarrow \delta_X(x_1,x_3) \leq \delta_X(x_1,x_2) + \delta_X(x_2,x_3)$
    \end{enumerate}
\end{definition}

\begin{definition}
    A lattice $(X,\leq,\sqcup)$ with a quasi-metric $\delta_X$ is called a \emph{lattice quasi-metric space}.
\end{definition}

In a lattice quasi-metric space, $\delta_X(x_1,x_2)$ can intuitively be described as how ``close'' $x_1$ and $x_2$ are in some natural sense. For example, in the lattice of sets of integers, we would naturally have the distance between $\{1\}$ and $\{1,2\}$ is smaller than the distance between $\{1\}$ and $\{1,\ldots,100\}$, or that the distance between $\{1\}$ and $\{1,2\}$ is smaller than the distance between $\{1\}$ and $\{1,1000\}$ (all depending on the context of the lattice). Observe that, while the length of the shortest path from $x_1$ to $x_2$ may work as a quasi-metric, this is not \emph{the} quasi-metric that one may want to use, as there may be paths in the lattice that should naturally be considered longer than others (in the sense of distance), even if the paths are the same length in the lattice.

We may also find it easier to work with more natural quasi-metrics, so we may often impose a requirement of a \emph{normal} quasi-metric.

\begin{definition}
    A quasi-metric $\delta_X$ over a lattice $(X,\leq,\sqcup)$ is said to be \emph{height-normal} (or simply, normal) if it satisfies
    \begin{itemize}
        \item $a\leq b\leq c \implies \delta_X(a,b),\delta_X(b,c) \leq \delta_X(a,c)$
    \end{itemize}
    for all $a,b,c \in X$.
\end{definition}

We now introduce a new version of the lattice agreement problem. The $\varepsilon$-\emph{bounded} lattice agreement problem (with $\varepsilon \geq 0$) is as follows. Let $(X,\leq,\sqcup)$ be a lattice quasi-metric space with $\delta_X$. Each process $p_i$ proposes an input value $x_i \in X$, and each $p_i$ decides on an output $y_i \in X$. To solve the $\varepsilon$-bounded lattice agreement, Downward-Validity, Upward-Valildity, and Comparability must all be met, along with one additional constraint:
\begin{itemize}
    \item {\bf $\varepsilon$-Tightness}: $\forall i,j \in [n] : (y_i \leq y_j \Rightarrow \delta_X(y_i,y_j) \leq \varepsilon)$
\end{itemize}
This forces the output values to be \emph{approximately} the same, i.e. ``close'' to each other. 

Later in the paper, we will be interested in studying the tightness of specific instances of the $\varepsilon$-bounded lattice agreement problem. For the following definitions, an instance of a protocol is a pair $(I,Y)$ of sets of input and output values, respectively. In context, $I = \{x_i \mid i \in [n]\}$ and $Y = \{y_i \mid i \in [n]\}$.

\begin{definition}
    An instance of a protocol of the lattice agreement problem is said to be $\gamma$-compliant if 
    \begin{enumerate}[(1)]
        \item The outputs satisfy $\gamma$-Tightness;
        \item There is no $\varepsilon < \gamma$ such that the outputs satisfy $\varepsilon$-Tightness.
    \end{enumerate}
\end{definition}

\begin{lemma}
    Any instance of any lattice agreement protocol is $\gamma$-compliant for a unique $\gamma \geq 0$ determined by $\gamma = \max_{i,j \in [n]} \delta_X(y_i,y_j)$. \label{gammaunique}
\end{lemma}
\begin{proof}
    Let $\gamma = \max_{i,j \in [n]} \delta_X(y_i,y_j)$ (which must satisfy $\gamma \geq 0$, since the outputs form a chain by Comparability). Then the outputs certainly satisfy $\gamma$-Tightness by definition, and if $\varepsilon < \gamma$, then the $i,j$ that maximize $\gamma$ (i.e. the $i,j$ such that $\gamma = \delta_X(y_i,y_j)$) do not satisfy $\delta_X(y_i,y_j) \leq \varepsilon$, so $\varepsilon$-Tightness is not satisfied, and so the instance is $\gamma$-compliant. Furthermore, this means the instance is not $\varepsilon$-compliant for any $\varepsilon < \gamma$. If $\varepsilon > \gamma$, the instance is clearly not $\varepsilon$-compliant by Condition 2. Thus, $\gamma$ is unique.
\end{proof}

We also prove that, for normal $\delta_X$, $\gamma$ is determined by the min and max outputs.

\begin{lemma}
    For any instance of the lattice agreement problem and a normal $\delta_X$, let $Y = \{y_i \mid i\in [n]\}$, and $\gamma = \delta_X(\min{Y}, \max{Y})$. Then, this instance is $\gamma$-compliant. \label{normgamma}
\end{lemma}
\begin{proof}
    First observe that $\min{Y},\max{Y}$ exists by Comparability. Let $y_i = \min{Y}$ and $y_j = \max{Y}$, and let $y_k,y_{\ell} \in Y$ such that $y_k \leq y_{\ell}$. Then $y_i \leq y_k \leq y_{\ell} \leq y_j$, so by normality of $\delta_X$, $\delta_X(y_i,y_j) \geq \delta_X(y_k,y_j) \geq \delta_X(y_k,y_{\ell})$. Thus, by Lemma \ref{gammaunique}, $\gamma = \delta_X(y_i,y_j) = \delta_X(\min{Y},\max{Y})$ (since $y_i,y_j$ have the maximum distance among the outputs as just shown).
\end{proof}

\subsection{Equivalences}



In the following lemmas we present equivalences between different instances of the $\varepsilon$-bounded lattice agreement problem and one result displaying an equivalence with the original lattice agreement problem.

An obvious equivalence is that a protocol which guarantees tighter bounds can also be used to solve a bounded lattice agreement problem for looser bounds.

\begin{lemma}
    Let $\varepsilon' \geq \varepsilon$. Then any protocol that solves the $\varepsilon$-bounded lattice agreement problem can be used to solve the $\varepsilon'$-bounded lattice agreement problem.\label{uppclosureequiv}
\end{lemma}
\begin{proof}
    Downward-Validity, Upward-Validity, and Comparability are all satisfied by the $\varepsilon$-bounded protocol. By $\varepsilon$-Tightness, $\delta_X(y_i,y_j) \leq \varepsilon$ for all $y_i \leq y_j$, and since $\varepsilon \leq \varepsilon'$, we have $\delta_X(y_i,y_j) \leq \varepsilon'$, so $\varepsilon'$-Tightness is also satisfied.
\end{proof}

The following lemma shows that the $\varepsilon$-bounded lattice agreement problem for a $\varepsilon$ less than the minimum distance between any pair of elements in the lattice is the same as the 0-bounded lattice agreement problem.

\begin{lemma}
    Let $\varepsilon > 0$ such that $x < y \implies \delta_X(x,y) > \varepsilon$. The $\varepsilon$-bounded lattice agreement problem is equivalent to the 0-bounded lattice agreement problem. \label{0equiv}
\end{lemma}
\begin{proof}
    Suppose there exists a protocol to solve the $\varepsilon$-bounded lattice agreement problem for such a $\varepsilon$. Then, by $\varepsilon$-Tightness, $\delta_X(y_i,y_j) \leq \varepsilon$ for all (correct) $i,j \in [n]$ with $y_i\leq y_j$.
    
    Suppose there is some $y_i \neq y_j$. By comparability, either $y_i \leq y_j$ or $y_j \leq y_i$. Suppose, wlog, the former is true. Then by our chosen $\varepsilon$, we have $\delta_X(y_i,y_j) > \varepsilon$, which contradicts $\varepsilon$-Tightness. Thus, there is no $y_i,y_j$ such that $y_i \neq y_j$, and so this protocol solves the 0-bounded lattice agreement problem.

    The other direction is trivial since $\varepsilon > 0$.\label{minbound}
\end{proof}

We also show another equivalence, in that the original lattice agreement problem is equivalent to the $\varepsilon$-bounded agreement problem when $\varepsilon \geq D$, where 
\begin{equation*}
    D = \max_{s_1,s_2 \bowtie \{x_i \mid i \in [n]\}} \delta_X(s_1,s_2).
\end{equation*}
Intuitively, $D$ is the longest distance between any two points in the range of values in which the outputs can lie.

\begin{lemma}
    The lattice agreement problem is equivalent to the $D$-bounded lattice agreement problem.\label{uppbnd1}
\end{lemma}
\begin{proof}
    Let $P$ be a protocol solving the lattice agreement problem. Then, by Downward-Validity $x_i \leq y_i$ for all $i \in [n]$. By Upward-Validity, $y_i \leq \bigsqcup_{j\in [n]} x_j$, so $y_i \bowtie \{x_j \mid j\in [n]\}$ for all $i \in [n]$, so $\delta_X(y_i,y_j) \leq D$ for all $i,j \in [n]$ by maximality of $D$.

    The other direction is trivial since Downward-Validity, Upward-Validity, and Comparability are all achieved in the $D$-bounded lattice agreement problem.
\end{proof}

If we have normality in $\delta_X$, however, then we may attain a stricter upper-bound on $\varepsilon$ for guaranteed by the lattice agreement problem. We will let
\begin{equation*}
    D' = \max_{i\in [n]} \,\delta_X(x_i,\bigsqcup_{j\in[n]} x_j).
\end{equation*}
and we first show that $D'$ is indeed a stronger upper-bound.

\begin{lemma}
    For an input set $\{x_i \mid i \in [n]\}$, $D' \leq D$.
\end{lemma}
\begin{proof}
    Since for all $i\in [n]$, $x_i \bowtie \{x_i \mid i \in [n]\}$ and $\bigsqcup_{j\in[n]} x_j \bowtie \{x_i \mid i \in [n]\}$, we immediately have $D' \leq D$.
\end{proof}

\begin{lemma}
    If $\delta_X$ is normal, then the lattice agreement problem is equivalent to the $D'$-bounded lattice agreement problem.
\end{lemma}
\begin{proof}
    Let $P$ be a protocol solving the lattice agreement problem. Let $x_i$ be an input which maximizes $D'$ (that is, $x_i \in \{x_j \mid j \in [n]\}$ such that $D' = \delta_X(x, \bigsqcup_{j\in [n]} x_j)$). Then, by Downward-Validity, $x_i \leq y_i$. Let $y_j$ be any output, then there are two cases by Comparability. If $y_i \leq y_j$, then $x_i \leq y_j$, so $\delta_X(y_i, y_j) \leq D'$ by normality. If $y_j \leq y_i$, then $\delta_X(y_j,y_i) \leq \delta_X(y_j,\bigsqcup_{j\in [n]} x_j) \leq \delta_X(x_j,\bigsqcup_{j\in [n]} x_j)$ by normality, and $\delta_X(x_j,\bigsqcup_{j\in [n]} x_j) \leq D'$ since $D'$ is maximized by $x_i$.


    Then, for any pair of outputs $y_i \leq y_j$ (or, wlog, $y_j \leq y_i$), we have $\delta_X(y_i,y_j) \leq D'$.

    The other direction is trivial as stated in Lemma \ref{uppbnd1}.
\end{proof}

These equivalences will help us reduce a seemingly continuous set of problems to an equivalent discrete set of problems that is much easier to reason about.

\subsection{Reconciliation Protocols}

Suppose we already have a set of output values $\{y_i \mid i \in [n]\}$ satisfying Downward-Validity, Upward-Validity, and Comparability with respect to some input set $\{x_i \mid i \in [n]\}$. Our goal is to achieve $\varepsilon$-Tightness with at most $f < n$ faults when possible. These algorithms that begin with outputs from standard lattice agreement will be called \emph{reconciliation} protocols, and they are correct w.r.t. $\varepsilon$ when Downward-Validity, Upward-Validity, and Comparability are all maintained in the new outputs, and $\varepsilon$-Tightness is also achieved in the new outputs. 

An instance $(I,Y,Y')$ of a reconciliation protocol then consists of an additional set $Y'$ representing the outputs from reconciliation on $Y$, which satisfies lattice agreement requirements w.r.t. $I$.

We show an important result that allows us to focus merely on reconciliation protocols for the sake of analyzing solvability of the $\varepsilon$-bounded lattice agreement problem.

\begin{theorem}
    There exists a $\varepsilon$-reconciliation protocol with $f$ faults if and only if there exists a protocol solving $\varepsilon$-bounded lattice agreement with $f$ faults.
\end{theorem}
\begin{proof}
    If there exists a $\varepsilon$-reconciliation protocol, then a protocol solving $\varepsilon$-bounded lattice agreement by simply running a standard lattice agreement protocol (e.g. synchronous and asynchronous algorithms presented in \cite{attiya95}), followed by the reconciliation protocol. The outputs from the lattice agreement must satisfy Downward-Validity, Upward-Validity, and Comparability, so they are valid inputs to the reconciliation protocol. By definition, the reconciliation protocol must maintain Downward-Validity, Upward-Validity, and Comparability, and it must also achieve $\varepsilon$-Tightness. Therefore, the $\varepsilon$-bounded lattice agreement is solved.

    In the other direction, simply run the $\varepsilon$-bounded lattice agreement protocol. Then, the output values satisfy everything that is needed for a $\varepsilon$-reconciliation protocol, so we are done. 
\end{proof}

This shows that it is sufficient to study reconciliation protocols for the sake of solvability of the $\varepsilon$-bounded lattice agreement problem, and vice versa. This is not the case, however, when attempting to minimize message complexity (as is clear from the proof).

\section{Synchronous Reconciliation}

Solvability in synchronous systems is very straightforward by a generalizable synchronous algorithm for binary consensus.

\begin{algorithm}[H]
\caption{Synchronous Agreement}
\label{alg:example}
\begin{algorithmic}[1]
\STATE $V \gets \{y_i\}$ \COMMENT{Initially just the lattice agreement output}
\FOR{$k = 1$ to $f+1$}
    \STATE Send $\{v \in V \mid P_i \text{ has not sent } v \text{ before}\}$ to all
    \FOR{each $j \in [n]$}
        \STATE Receive $S_j$ from $P_j$
        \STATE $V \gets V \cup S_j$
    \ENDFOR
\ENDFOR
\STATE Decide on $\max(V)$
\end{algorithmic}
\end{algorithm}

\begin{theorem}
    All correct processes decide on the same value. \label{synccorrect}
\end{theorem}
\begin{proof}
    Since there are $f+1$ rounds and at most $f$ faults, there must (by pigeonhole principle) be a round in which no faults occur. That is, all correct processes send their set of unsent values to every other process, and all correct processes receive the respective sets from all other correct processes. Thus, at the end of that round, all processes must receive the same values, and end with the same $V$. Then, $V$ stays the same for the remaining rounds, so all correct processes finish by deciding on the same value.
\end{proof}

\begin{corollary}
    The synchronous $\varepsilon$-bounded lattice agreement problem is solvable for all $\varepsilon \geq 0$.
\end{corollary}
\begin{proof}
    Immediate from Theorem \ref{synccorrect} and Lemma $\ref{uppclosureequiv}$.
\end{proof}

\section{Asynchronous Reconciliation}

While we can solve the $\varepsilon$-bounded lattice agreement problem for all $\varepsilon \geq 0$ synchronously, we certainly cannot solve the 0-bounded lattice agreement problem asynchronously.

\begin{theorem}
    The asynchronous 0-bounded lattice agreement problem is unsolvable for $f \geq 1$. \label{eq0unsolv}
\end{theorem}
\begin{proof}
    Suppose there exists a protocol to solve the asynchronous 0-bounded lattice agreement problem. We show that such a protocol would solve the FLP asynchronous binary consensus problem presented in \cite{flp}. 
    
    Then, it certainly terminates since every process eventually decides on an output. Let the lattice be $\{0,1\}$ with $0 < 1$. Then, if $x_i = 0$ for all $i \in [n]$, $\bigsqcup_{j\in [n]} x_j = 0$, and so $0 \leq y_i \leq \bigsqcup_{j\in [n]} x_j = 0$ by Downward and Upward-Validity (and similarly if $x_i = 1$ for all $i \in [n]$). Thus we satisfy the Validity requirement of FLP. 
    
    Agreement is guaranteed by 0-Tightness, so this protocol guarantees FLP conditions. Since there may be a crash fault ($f \geq 1$), this is a contradiction by the FLP result \cite{flp}.
\end{proof}

In fact, by Lemma \ref{0equiv}, we cannot solve this problem asynchronously for $\varepsilon < M$ where
\begin{equation*}
    M = \min{\{\delta_X(x,y) \mid x,y \in X, x\leq y\}}
\end{equation*}
While this gives a lower bound for $\varepsilon$, the arbitrary nature of the chosen quasi-metric gives the following fundamental result.

\begin{theorem}
    There does not exist an asynchronous protocol that always solves the $\varepsilon$-bounded lattice agreement problem for any $\varepsilon > 0$. \label{geq0unsolv}
\end{theorem}
\begin{proof}
    It is easy to construct a lattice and set of inputs for which all elements in the join-closed sub-lattice generated by the inputs are all more than $\varepsilon$ units apart. Here, $\varepsilon$-Tightness is certainly impossible to achieve.
\end{proof}

Observe that Theorem \ref{geq0unsolv} (and the following Corollary) is over all lattices, however it does not say anything about existence of protocols for \emph{specific} lattices. Soon, we will prove an impossibility result on the upper-bound depending on not only the lattice, but the particular instance in the lattice as well. 

\begin{corollary}
    The asynchronous $\varepsilon$-bounded lattice agreement problem is unsolvable for all $\varepsilon \geq 0$ with $f \geq 1$.
\end{corollary}
\begin{proof}
    Immediate from Theorems \ref{eq0unsolv} and \ref{geq0unsolv}.
\end{proof}

This result shows us that, even for viable values of $\varepsilon$, there is no protocol relying on $\varepsilon$ that can always guarantee a valid set of outputs per the problem. Furthermore, this shows that we cannot make an algorithm for arbitrary tightness stronger than the upper-bound given by standard lattice agreement (for arbitrary lattices and quasi-metrics). Therefore, we instead present algorithms for achieving better $\varepsilon$-Tightness, but not a particular value of $\varepsilon$. That is, we would like to attain stronger \emph{upper-bounds} on compliance.



\begin{theorem}
    There is no asynchronous reconciliation protocol that guarantees $\varepsilon$-Tightness for any fixed $\varepsilon < D'$ when $f \geq 1$. \label{uppboundimposs}
\end{theorem}
\begin{proof}

    Let $Y = \{y_i \mid i\in [n]\}$ and suppose
    \begin{equation*}
        \delta_X(\min{Y},\max{Y}) = \max_{i\in [n]} \,\delta_X(x_i,\bigsqcup_{j\in[n]} x_j) = D'.
    \end{equation*}
    Then, it must be the case that such a protocol changes the value of $\min{Y}$ or $\max{Y}$ in their respective processes. Suppose $\min{Y} = y_1 = \ldots = y_{n-1} = x_1 = \ldots = x_{n-1}$ and $\max{Y} = y_n = x_n = \bigsqcup_{j \in [n]} x_j$. 

    We cannot move $y_n = \max{Y}$ to any other point in the lattice since either Downward or Upward-Validity would be violated.

    We cannot move any $y_1,\ldots,y_{n-1}$ down in the lattice, since then Downward-Validity would be violated. Then, to decrease $\varepsilon$ by changing $\min{Y}$, we must move all $y_1,\ldots,y_{n-1}$ up the lattice. 
    
    Since the system is asynchronous, we may repeatedly delay $p_n$'s message, meaning that transitions in the state of values are made without depending on $y_n$. Then, if some process $p_i$ changes their value $y_i$ to a different $y_i'$, then this transition could also occur if $x_n = y_n = y_1 = \ldots$. This would then violate Upward-Validity.

    Therefore, an adversary in this situation can always enforce a transition to an illegal state by delaying $p_n$'s messages.
\end{proof}

While Theorem \ref{uppboundimposs} seems similar to the previous unsolvability results, there is an important distinction. This theorem states that we cannot even attain a guarantee-able \emph{upper-bound} on $\gamma$-compliance better than $D'$.

This motivates a heuristic approach to $\varepsilon$-bounded lattice agreement, where we would like to attain such an upper-bound (or something even stronger) on average or with high probability.

\subsection{Heuristic Algorithm}

If we have a normal quasi-metric, then improving the $\gamma$-compliance from $\max_{i,j \in [n]} \delta_X(y_i,y_j)$ (by Lemma \ref{gammaunique}) is simply a matter of either moving all minimum values up the lattice or moving all the maximum values down the lattice (while maintaining validity).

Moving maximum values down can certainly cause us to violate Downward-Validity, however, since the original input $x_i$ of a maximum value $y_i$ could be equal to it (that is, $x_i = y_i$). Moving minimum values up is therefore a much better option, since we only need to stay below $\bigsqcup_{j \in [n]} x_j$. Without knowing the inputs, however, this may be challenging. The only thing we do know in a reconciliation protocol is that all values satisfy Upward-Validity, so we can rely on these values.

Here, we present a very simple algorithm for attaining better $\gamma$-compliance with high probability (depending on the chosen value of $k$, and of course the set of initial values).

\begin{algorithm}[H]
\caption{${\rm DR}(k)$}
\label{alg:heuristicalgo}
\begin{algorithmic}[1]
\STATE $y_i \gets \text{Output from lattice agreement}$
\FOR{$r = 1$ to $k$}
    \STATE Send $(y_i,r)$ to all
    \STATE Receive $n-f$ $(\cdot, r)$ messages. Let $V_r$ be the set of received values.
    \STATE $y_i \gets \max{(V_r \cup \{y_i\})}$
\ENDFOR
\STATE Output $y_i$
\end{algorithmic}
\end{algorithm}

Let $y_i^r$ be the value of $y_i$ in process $p_i$ before executing the for loop for round $r$. Then, the initial value (output from lattice agreement) is $y_i^1$, and we let $y_i^{k+1}$ be the final output value. Then, let $A_r = \{y_i^r \mid i \in [n]\}$.

\begin{lemma}
    For all $1 \leq m \leq k$, $A_{m+1} \subseteq A_m$. \label{containprop}
\end{lemma}

\begin{theorem}
    Downward-Validity, Upward-Validity, and Comparability are all maintained.
\end{theorem}
\begin{proof}
    For each $y_i^r$, we have $y_i^r \leq y_i^{r+1}$, and since $x_i \leq y_i^1$, we then have $x_i \leq y_i^{k+1}$, so Downward-Validity is maintained.

    By Lemma \ref{containprop}, we have that $A_{k+1} \subseteq A_1$, so for all $i\in [n]$, $\exists j \in [n]$ such that $y_i^{k+1} = y_j^1$. Since the initial values all satisfy Upward-Validity, we then have that $y_i^{k+1}$ satisfies Upward-Validity. This same argument also follows to show that Comparability is maintained as well.
\end{proof}

\begin{theorem}
    For any instance $(I,Y,Y')$ of this algorithm, if $(I,Y)$ is $\gamma$-compliant, then $(I,Y')$ is $\gamma'$-compliant for some $\gamma' \leq \gamma$.
\end{theorem}
\begin{proof}
    By Lemma \ref{containprop}, the maximum distance between any pair of values can only decrease from $A_{m}$ to $A_{m+1}$.
\end{proof}

Observe that none of these proofs uses the normality of $\delta_X$. Therefore, the algorithm still works for non-normal quasi-metrics, however there is a much greater likelihood of improving $\gamma$-compliance with a normal quasi-metric since we can guarantee improvement when minimum values are moved up.

Simulating this algorithm for large $n$ is computationally difficult, so we instead provide an easier to compute non-distributed abstract semantics that gives a good approximation of the improvement rates with ${\rm DR}(k)$.

\subsection{Simplified Approximate Model}

Since, for all $i \in [n]$, $y_i^r \leq y_i^{r+1}$, we only need the minimum values to change in order to attain a $\gamma$-compliance better than $D'$. This gives rise to a simple model that allows us to analyze any given instance of ${\rm DR}(k)$ as a discrete-time absorbing Markov chain.

The idea is to represent states by 0s and 1s, where a 0 represents a minimum value (that is, some $y_i^r$ such that $y_i^r = \min{\{y_j^1 \mid j \in [n]\}}$) and 1 represents a non-minimum value. The goal is to have all 1s and no 0s, so that there are no more minimum values (with respect to the initial values).

At each round of the algorithm, each 0 picks $n-f$ of the current state's 0s and 1s, and changes to a 1 only if it picked at least 1 (that is, it takes the max). 

For any given process, we assume that the probability of receiving a message from another process on round $r$ is uniformly distributed. That is, there is an equal chance for $p_i$ to receive from $p_j$ for all $j \in [n]$ at each round. Thus, in the simplified model, this means that each bit has an equal chance of being chosen by each 0 independently.

We also assume that crashes are independent, and each process has an equal chance $p_f$ to crash (among the first $f$ crashes). To simplify this model, we say that line 3 of ${\rm DR}$ is atomic, so processes cannot send a message to some processes and crash before sending it to the rest of the processes.

\begin{definition}
    The \emph{state space} $\mathbb{S}_{n,f}$ (or simply $\mathbb{S}$ when $n,f$ are clear from context) for $n$ processes and at most $f$ failures consists of all vectors $\langle s_1,\ldots,s_n \rangle$ such that $s_i \in \{0,1,\bot\}$ for all $i \in [n]$, and there are no more than $f$ components equal to $\bot$ and at least one component equals 1.
\end{definition}

For convenience, we say that $\bot < 0 < 1$.

\begin{definition}
    A state $S' = \langle s_1',\ldots,s_n' \rangle \in \mathbb{S}$ is \emph{reachable} from state $S = \langle s_1,\ldots,s_n \rangle \in \mathbb{S}$ (written $S_1 \rightarrow S_2$) if
    \begin{itemize}
        \item $\forall i \in [n] : s_i = \bot \implies s_i' = \bot$
        \item $\forall i \in [n] : s_i = 1 \implies s_i' \neq 0$.
    \end{itemize}
\end{definition}

This definition corresponds to the reachable states after running one round of ${\rm DR}$.

We now define the set of ``good'' states, in which compliance was improved.

\begin{definition}
    The set of improved states $\mathbb{F} \subseteq \mathbb{S}$ is the set $\mathbb{F} = \{\langle s_1,\ldots,s_n \rangle \in \mathbb{S} \mid (\forall i \in [n] : s_i \neq 0) \vee (\forall i \in [n] : s_i \neq 1) \}$.
\end{definition}

We now provide an algorithm for the abstract semantics with $k$ rounds beginning at state $S$, and traversing the state reachability graph (as determined by the random $n-f$ choices).

\begin{algorithm}
\caption{Simplified Approximate Model Simulation}
\begin{algorithmic}[1]
\STATE $x \gets 0$
\FOR{$r \gets 1$ to $k$}
    \STATE $A \gets \text{copy of $S$}$
    \FOR{$i \gets 1$ to $n$}
        \IF{$x < f$ and $p_f$ probability}
                \STATE $S[i] \gets \bot$
                \STATE $x \gets x + 1$
                \STATE \textbf{continue}
        \ENDIF
        \IF{$S[i] \neq 0$}
                \STATE \textbf{continue}
        \ENDIF
        \STATE $C \gets \text{choose $n-f$ random values from $S$}$
        \STATE $A[i] \gets \max(C \cup \{S[i]\})$
    \ENDFOR
    \STATE $S \gets A$
\ENDFOR
\STATE \textbf{return} $S \in \mathbb{F}$?
\end{algorithmic}
\end{algorithm}

This algorithm was simulated for 1000 runs with $n=1000$. The following tables show results for when $f=200$ and $p_f=0.06$. Table \ref{rand1} gives success rates over the 1000 runs with randomized inputs. Table \ref{wc1} gives success rates with worst-case input, which is all 0s and one 1.

\begin{table}[htbp]
\caption{Success rates of random input with $n=1000,p_f=0.06$}
\begin{center}
\begin{tabular}{|c|c|c|c|c|}
\hline
     & $k=2$ & $k=3$ & $k=4$ \\\hline
    $f=200$ & 17.1\% & 90.3\% & 99.9\% \\\hline
    $f=800$ & 16.8\% & 89.6\% & 99.3\% \\\hline
\end{tabular}
\label{rand1}
\end{center}
\end{table}

\begin{table}[htbp]
\caption{Success rates of worst-case input with $n=1000,p_f=0.06$}
\begin{center}
\begin{tabular}{|c|c|c|c|c|}
\hline
     & $k=2$ & $k=3$ & $k=4$ & $k=5$ \\\hline
    $f=200$ & 0.0\% & 41.3\% & 97.6\% & 100.0\% \\\hline
    $f=800$ & 0.0\% & 5.4\% & 84.0\% & 98.9\% \\\hline
\end{tabular}
\label{wc1}
\end{center}
\end{table}

We also tested $p_f = 0.5$ to $p_f=0.8$ in 0.1 for $f=800$ and observed that, for $k=2$, 0\% of simulations succeeded in achieving a better upper-bound, however for $k=3$, 100\% of simulations succeeded. This is to be expected, since using more of the fault budget allows either all 1s to be eliminated, or enough 0s to be eliminated so that a 1 is guaranteed among $n-f$ choices.

Overall, we see that only a very small number of rounds in relation to the number of processes is needed for very high probability to improve the upper-bound. For 1000 processes, the average case (randomized input) gives that we only need to run 4 rounds to get a high probability of success (5 for a near guarantee, even in worst-case input), which is nearly negligible compared to the number of rounds needed for the lattice agreement itself. 

We conjecture that the same probability of improvement is the same at $k=5$ for any number of processes $n$, failures $f$, and chance of failure $p_f$, with minimal error. That is, we conjecture that such a high probability of improvement is achieved with ${\rm DR}(5)$ in any system (under the assumptions of our original system model). This is because the same amount ``spreading'' per round will occur proportional to the number of processes, causing it to converge at the same round, even if the first few rounds differ when $n,f,p_f$ differ. 

Furthermore, this would mean that our heuristic algorithm only needs $O(1)$ (constant) rounds to achieve improved $\gamma$-compliance, and the constant number of rounds is very small (5, as per the conjecture).

\section{Conclusion}

In this paper, we proposed a new constraint on the standard lattice agreement problem by requiring outputs to be ``close'' to each other w.r.t. a quasi-metric on the lattice. We show that this new problem whose range is over $\Ri$ can be reduced to a discrete set of equivalence classes on this range. An upper-bound on tightness is derived for standard lattice agreement, and we show that these upper-bounds are stubborn in the presence of asynchrony. For synchronous systems, we show that the problem is easily solvable for all desired levels of tightness. For asynchronous systems, we provide a variety of impossibility results, and one which states that the upper-bound achieved in standard lattice agreement cannot be improved with certainty. This led us to present a heuristic algorithm for attaining improved tightness with high probability in asynchronous systems. We then modeled this heuristic algorithm with an approximation and simulated this model to obtain an approximation on the probabilities of improving tightness with $k$ rounds. Finally, we conjectured that only a constant number ($k=5$) of rounds are necessary to achieve high probability of improvement in the presence of any number of processes, faults ($f < n$), and any fault probability.

\bibliographystyle{plain}
\bibliography{refs}

\end{document}